\documentclass[a4paper,UKenglish,cleveref, numberwithinsect, thm-restate]{lipics-v2021}

\usepackage{prooftree}
\usepackage{mathpartir}
\usepackage{xcolor}
\usepackage{listings}
\usepackage{tikz-cd}
\usepackage{stmaryrd}
\usepackage{amsmath}
\usepackage{todonotes}
\usepackage{footnote}
\usepackage{url}
\usepackage{wasysym}
\usepackage{hyperref}
\usepackage{mpst_macros}

\definecolor{colorsec}{HTML}{345A8A}
\definecolor{colorsubsec}{HTML}{4F81BD}
\definecolor{colorsubsubsec}{HTML}{5388C8}

\hypersetup{
  colorlinks = true,
  citecolor=purple,
  linkcolor=colorsubsec,
}

\definecolor{bg}{rgb}{0.95,0.95,0.95}

\makeatletter

\usepackage[framemethod=tikz]{mdframed}
\usepackage[most]{tcolorbox}

\def\beginlstdelim#1#2#3#4%
{%
  \def\endlstdelim{#2\egroup}%
  {\ttfamily#3#1}\bgroup#4\aftergroup\endlstdelim%
}

\lstdefinelanguage{smt}{
  language=lisp,
  alsoletter=0123456789>=,
  keywords={define-fun,declare-fun,declare-const,set-option,echo,set-info,exit,pop,push,assert},
  classoffset=1,
  morekeywords={Int,Real,Bool},keywordstyle=\color{blue!60!black}\bfseries,
  classoffset=2,
  morekeywords={not,and,or,=>},keywordstyle=\bfseries,
  classoffset=3,
  morekeywords={check-sat,check-sat-assume},keywordstyle=\color{red}\bfseries,
  classoffset=4,
  morekeywords={true,false,0,1,2,3,4,5,6,7,8,9,10},keywordstyle=\color{violet}\bfseries,
  classoffset=0,
  moredelim=**[is][\beginlstdelim{define-fun\ }{\ }{\color{green!50!black}\bfseries}{\color{blue!80!black!50!white}\bfseries}]{define-fun\ }{\ },
  moredelim=**[is][\beginlstdelim{declare-const\ }{\ }{\color{green!50!black}\bfseries}{\color{blue!80!black!50!white}\bfseries}]{declare-const\ }{\ },
  moredelim=**[is][\beginlstdelim{declare-fun\ }{\ }{\color{green!50!black}\bfseries}{\color{blue!80!black!50!white}\bfseries}]{declare-fun\ }{\ },
  mathescape=false
}

\lstdefinelanguage{Rocq}{
  alsoletter={<}{:}0123456789,
  morekeywords={Variable,Variant,Inductive,CoInductive,Fixpoint,CoFixpoint,%
    Definition,Program, Lemma,Theorem,Corollary,Axiom,Local,Save,Grammar,Syntax,Intro,%
    Trivial,Qed,Intros,Symmetry,Simpl,Rewrite,Apply,Elim,Assumption,%
    Left,Cut,Case,Auto,Unfold,Exact,Right,Hypothesis,Pattern,Destruct,%
    Constructor,Defined,Fix,Record,Proof,Induction,Hints,Exists,%
    Parameter,Split,Red,Reflexivity,Transitivity,if,then,else,Opaque,Module,%
    Transparent,Inversion,Absurd,Generalize,Mutual,Cases,of,Analyze,%
    AutoRewrite,Functional,Scheme,params,Refine,using,Discriminate,Try,%
    Require,Load,Import,Scope,Open,Section,End,Ltac,fun,forall,exists,Canonical,Structure,Eval,Notation,as,return,Goal,Class,Module%
  },%
  classoffset=1,
  morekeywords={Type,Prop,bool,nat,Set,let,in,match,with,end,as,<:,Z,farray,bitvector},keywordstyle=\color{blue!60!black}\bfseries,
  classoffset=2,
  morekeywords={Error:,Warning:},keywordstyle=\color{red}\bfseries,
  classoffset=3,
  morekeywords={0,1,2,3,4,5,6,7,8,9,10,11,12,13,14,15,16,16,18,19,20},keywordstyle=\color{violet},
  classoffset=0,
  sensitive, %
  moredelim=**[is][\beginlstdelim{CoInductive\ }{\ }{\color{green!50!black}\bfseries}{\color{blue!80!black!50!white}\bfseries}]{CoInductive\ }{\ },
  moredelim=**[is][\beginlstdelim{Inductive\ }{\ }{\color{green!50!black}\bfseries}{\color{blue!80!black!50!white}\bfseries}]{Inductive\ }{\ },
  moredelim=**[is][\beginlstdelim{Definition\ }{\ }{\color{green!50!black}\bfseries}{\color{blue!80!black!50!white}\bfseries}]{Definition\ }{\ },
  moredelim=**[is][\beginlstdelim{Lemma\ }{\ }{\color{green!50!black}\bfseries}{\color{blue!80!black!50!white}\bfseries}]{Lemma\ }{\ },
  moredelim=**[is][\beginlstdelim{Axiom\ }{\ }{\color{green!50!black}\bfseries}{\color{blue!80!black!50!white}\bfseries}]{Axiom\ }{\ },
  moredelim=**[is][\beginlstdelim{Theorem\ }{\ }{\color{green!50!black}\bfseries}{\color{blue!80!black!50!white}\bfseries}]{Theorem\ }{\ },
  moredelim=**[is][\beginlstdelim{Class\ }{\ }{\color{green!50!black}\bfseries}{\color{blue!60!black}\bfseries}]{Class\ }{\ },
  moredelim=**[is][\beginlstdelim{Module\ }{\ }{\color{green!50!black}\bfseries}{\color{blue!60!black}\bfseries}]{Module\ }{\ },
  moredelim=**[is][\beginlstdelim{Record\ }{\ }{\color{green!50!black}\bfseries}{\color{blue!60!black}\bfseries}]{Record\ }{\ },
  morecomment=[n]{(*}{*)},%
  morestring=[d]",%
  literate={=>}{{$\Rightarrow$}}1
  {->}{{$\,\to\,$}}1
  {<-}{{$\leftarrow$}}1
  {>->}{{$\rightarrowtail$}}2
  {<->}{{$\leftrightarrow$}}1
  {forall}{{\color{blue!60!black}\bfseries$\forall$}}1
  {exists}{{\color{blue!60!black}\bfseries$\exists$}}1
  {<>}{{$\neq$}}1
  {<=}{{$\leq$}}1
  {>=}{{$\geq$}}1
  {:=}{{$\triangleq$}}1
  {\/\\}{{$\wedge$}}1
  {|-}{{$\vdash$}}1
  {\\\/}{{$\vee$}}1
  {'}{'}1
  {⟦}{{$\llbracket$}}1
  {⟧}{{$\rrbracket$}}1
  {-->}{{$\longrightarrow$}}1
}

\lstdefinelanguage{lfsc}{
  language=lisp,
    alsoletter={!}{\%}{@}{\\},
  keywords={check,define,declare,program},
  classoffset=1,
  morekeywords={int,mpz,th_holds,holds,term,sort,type,match,fail,default},keywordstyle=\color{blue!60!black}\bfseries,
  classoffset=2,
  keywords={\%,@,!,\\},keywordstyle=\color{violet}\bfseries,
  classoffset=0,
  moredelim=**[is][\beginlstdelim{define\ }{\ }{\color{green!50!black}\bfseries}{\color{green!50!black!50!white}\bfseries}]{define\ }{\ },
  moredelim=**[is][\beginlstdelim{declare\ }{\ }{\color{green!50!black}\bfseries}{\color{blue!80!black!50!white}\bfseries}]{declare\ }{\ },
  moredelim=**[is][\beginlstdelim{program\ }{\ }{\color{green!50!black}\bfseries}{\color{blue!80!black!50!white}\bfseries}]{program\ }{\ },
  escapechar=\&
}

\lstdefinelanguage{smtcoq}{
  alsoletter=\#0123456789\=,
  classoffset=0,
  keywords={or,and,not,impl,true,false,\=,->},keywordstyle=\color{black}\bfseries,
  classoffset=1,
  morekeywords={0,1,2,3,4,5,6,7,8,9,10,11,12,13,14,15,16,16,18,19,20},keywordstyle=\color{violet}\bfseries,
  classoffset=0,
  sensitive=true,
  moredelim=**[is][\beginlstdelim{:(}{\ }{}{\color{green!50!black}\bfseries}]{:(}{\ },
}

\lstdefinelanguage{ocaml}{
  language=[Objective]caml,
  identifierstyle=\ocidstyle
}

\newcommand*\ocidstyle{%
        \expandafter\id@style\the\lst@token\relax
}
\def\id@style#1#2\relax{%
        \ifcat#1\relax\else
                \ifnum`#1=\uccode`#1%
                        \color{blue!60!black}
                \fi
        \fi
}

\newenvironment{tcb}[2][\footnotesize]{%
  \tcblisting{enhanced jigsaw,lines before break=3,
    listing only,colback=bg,colframe=bg,enlarge
    top by=0mm,top=5pt,bottom=0pt,left=2pt,right=2pt,enhanced,
    before={\vspace{10pt}},
    after={\par\vspace{5pt}\noindent},
    listing options={language=#2,basicstyle={\ttfamily#1\upshape}}%
    }}{\endtcblisting}

\newenvironment{tcbbr}[2][\footnotesize]{%
  \tcblisting{enhanced jigsaw,breakable,lines before break=3,
    listing only,colback=bg,colframe=bg,enlarge
    top by=0mm,top=5pt,bottom=0pt,left=2pt,right=2pt,enhanced,
    before={\vspace{10pt}},
    after={\par\vspace{5pt}\noindent},
    listing options={language=#2,basicstyle={\ttfamily#1\upshape}}%
    }}{\endtcblisting}

\definecolor{bgcolor}{HTML}{E0E0E0}

\newcommand{\code}[1]{\colorbox{bg}{\lstinline!#1!}}

\newcommand{\lstin}[1]{\lstinline|#1|}
\RenewCommandCopy{\theHtheorem}{\thetheorem}
\RenewCommandCopy{\theHlemma}{\thelemma}
\RenewCommandCopy{\theHexample}{\theexample}
\RenewCommandCopy{\theHdefinition}{\thedefinition}
\RenewCommandCopy{\theHremark}{\theremark}

\usepackage{silence}   
\newcommand{\ltgc}{\ltc}
\newcommand{\revised}[1]{\textcolor{black}{#1}}
\newcommand{\ltc}{\revised{local type environment}} 
\newcommand{\Ltc}{\revised{Local type environment}}
\newcommand{\LTC}{\revised{Local Type Environment}}
\newcommand{\LTgC}{\LTC}
\newcommand{\ctext}{\revised{environment}}
\newcommand{\Ctext}{\revised{Environment}}
\title{Formally Verified Liveness with Multiparty Session Types in Rocq}

\titlerunning{Formally Verified Liveness with Multiparty Session Types in Rocq}

\author{Omer Keskin}{University of Edinburgh, UK}{O.S.Keskin@SMS.ed.ac.uk}{https://orcid.org/0009-0002-7197-5158}
{}
\author{Nobuko Yoshida}
{University of Oxford, UK}{
nobuko.yoshida@cs.ox.ac.uk 
}{
https://orcid.org/0000-0002-3925-8557
}{EPSRC EP/T006544/2 , EP/T014709/2,
EP/Y005244/1, EP/V000462/1, EP/X015955/1, EU Horizon 101093006 and UKRI 10066667, Advanced Research and Invention Agency (ARIA), EP/Z533749/1 and a grant from the Simons Foundation.}
\author{Rob van Glabbeek}{University of Edinburgh, UK\newline
School of Computer Science and Engineering, University of New South Wales, Sydney,  Australia \and \url{https://theory.stanford.edu/~rvg/}} {rvg@cs.stanford.edu}{https://orcid.org/0000-0003-4712-7423}{Supported by Royal Society Wolfson Fellowship RSWF\textbackslash R1\textbackslash 221008}

\supplementdetails[subcategory={Source}]{Software}{\rocqrepo}

\authorrunning{O. Keskin, N. Yoshida and R.\,J. van Glabbeek}

\Copyright{Omer Keskin, Nobuko Yoshida and Rob van Glabbeek}



\category{}

\relatedversion{}

\acknowledgements{We deeply thank Burak Ekici for his collaborations, guidance and detailed feedback. We also thank the ITP'26 reviewers for the detailed feedback.}

\nolinenumbers

\hideLIPIcs  

\EventEditors{Ekaterina Komendantskaya and Tobias Nipkow}
\EventNoEds{2}
\EventLongTitle{17th International Conference on Interactive Theorem Proving (ITP 2026)}
\EventShortTitle{ITP 2026}
\EventAcronym{ITP}
\EventYear{2026}
\EventDate{July 26--29, 2026}
\EventLocation{Lisbon, Portugal}
\EventLogo{}
\SeriesVolume{382}
\ArticleNo{4}

\begin{document}

\maketitle

\begin{abstract}
Multiparty session types (MPST) offer a framework for the description of communication-based
protocols involving multiple participants. In the \textit{top-down} approach to MPST, 
the communication pattern of the session is described using a \textit{global type}. Then 
the global type is \textit{projected} on to a \textit{local type} for each participant,
and the individual processes making up the session are type-checked against these projections.
Typed sessions possess certain desirable properties such as \textit{safety}, \textit{deadlock-freedom} and 
\textit{liveness}.

In this work, we present the first mechanised proof of liveness
for synchronous multiparty session types in the Rocq Proof Assistant. 
Building on recent work, we represent global and local types
as coinductive trees using the Paco library. We use a coinductively defined \textit{subtyping} relation 
on local types together with another coinductively defined \textit{plain-merge} projection
relation relating local and global types.
We then \textit{associate} collections of local types, or \textit{{\ltc}s}, with 
global types using these projection and subtyping relations, and prove an \textit{operational correspondence}
between a {\ltc} and its associated global type. We utilise this
association relation to prove the safety and liveness of associated {\ltc}s
and, consequently, the multiparty sessions typed by these {\ctext}s.  

Besides clarifying the often informal proofs found in the MPST literature,
our Rocq mechanisation also enables the certification of liveness properties
of communication protocols. Our contribution amounts to around 14K
lines of Rocq code, 
available at \url{\rocqrepo}.

\end{abstract}

\newcommand{\rocqlink}[1]{\href{#1}{\includegraphics[width=11pt]{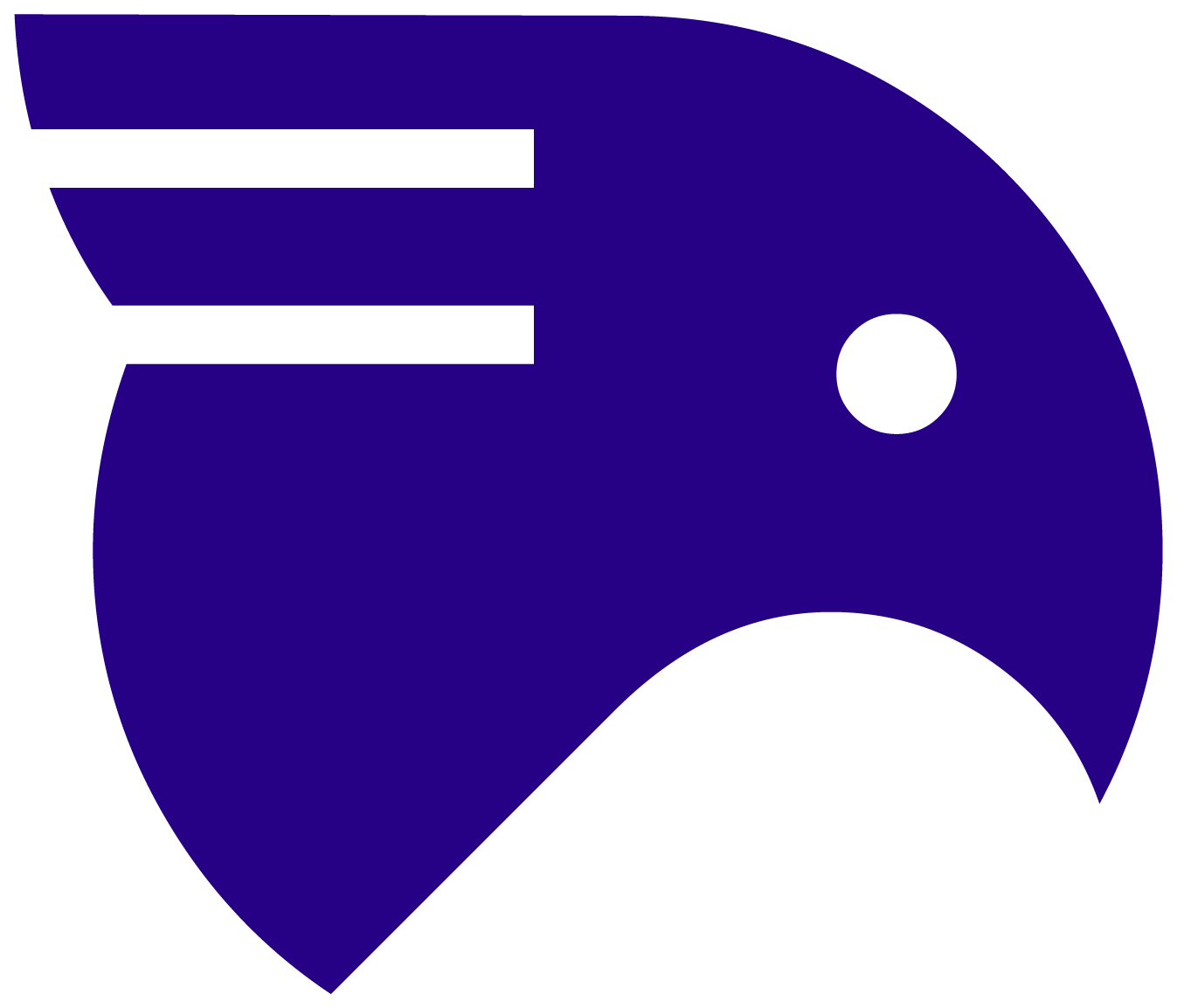}}}

\section{Introduction}
\label{sec:introduction}
Multiparty session types \cite{honda2008} provide a type discipline for the correct-by-construction specification of message-passing protocols.
Desirable protocol properties guaranteed by session types include \textit{communication safety} 
(the labels and types of senders' payloads cohere with the capabilities of the receivers),
\textit{deadlock-freedom} (also called \textit{progress} or \textit{non-stuck property} \cite{srpaper}) 
(it is possible for the session to progress so long as it has at least one active participant), and 
\textit{liveness} (also called \textit{lock-freedom} \cite{fairnesslock} or \textit{starvation-freedom} \cite{castro2026synthetic})   
(if a process is waiting to send or receive then a communication involving it eventually happens).

There exists two common methodologies for multiparty session types. 
In the \textit{bottom-up} approach, the individual processes making up the session are typed using 
a collection of \textit{participants} and \textit{local types}, that is, a \textit{{\ltc}}, and the properties of the session are examined by model-checking
this {\ltc}. Contrastingly, in the \textit{top-down} approach sessions are typed by a \textit{global type}
that is related to the processes using endpoint \textit{projections} and \textit{subtyping}.
The structure of the global type ensures that the desired properties are satisfied by the session.
These two approaches have their advantages and disadvantages: the bottom-up approach is generally 
able to type more sessions, while type-checking and 
type-inference in the top-down approach tend to be more efficient than model-checking 
the bottom-up system \cite{projsurvey}. 

\revised{
In this work, we present the Rocq \cite{coq} formalisation of a 
synchronous MPST system \cite{SynchronousSubtyping,VeryGentle} that ensures the aforementioned properties
for typed sessions. Our system types ($\vdash$) \textit{sessions} with 
\emph{{\ltc}s}.
Our local types closely mimic the structure of the processes; however they are not enough
to guarantee the properties we desire. To that end, we target typability 
by a certain set of well-behaved {\ltc}s that are \textit{associated} with
a global type.}   

\revised{
The \textit{association} ($\assoc$) \cite{LessIsMoreRevisited,PY2025} relation relates {\ltc}s and global types
using (coinductive plain) projection \cite{tirore2023sound} and \emph{subtyping}. 
It ensures \textit{operational correspondence}
between the labelled transition system (LTS) semantics we define 
for {\ltc}s and global types. By exploiting the structure of the global type and
relating the transition semantics for sessions, 
{\ltc}s and global types, we show that 
if an associated {\ltc} $\Gamma$ types a session $\M$, 
then $\M$ is guaranteed to possess
safety (\cref{theo-type-safe}), deadlock-freedom (\cref{theo-progress}) and 
liveness (\cref{theo-sess-live}).
}

\revised{
To our best knowledge, this work presents the first mechanisation of
liveness for multiparty session types in a proof assistant.
It is also the first work which mechanises
a MPST system
based on the association relation, and
proves its correctness by formalising 
properties of {\ltc}s $\Gamma$ 
(such as safety and deadlock-freedom). }

\begin{figure}
  \begin{minipage}{0.45\textwidth}
\small    
\begin{tikzcd}
\G \arrow[d, "\projf", dotted] \arrow[rd, "\assoc"] \arrow[rr, "\rightarrow",dotted] &   & \G' \arrow[d, "\assoc"]        \\
\T \arrow[d, "\vdash_{\mathsf{P}}", dotted] \arrow[r]                                  & {\{(\pp_i,\T_i) \; |\; i \in I\}=\Gamma} \arrow[d, "\vdash"] \arrow[r, "\rightarrow"] & \Gamma' \arrow[d, "\vdash"] \\
\mathsf{P} \arrow[r]                                                       & {\prod_{i}^{} \pp_i \triangleleft \mathsf{P}_i = \M} \arrow[r, "\rightarrow"]                                                     & \M'                          
\end{tikzcd}
\end{minipage}
\begin{minipage}{0.45\textwidth}
$\T$ refers to a local type, 
\G \;a global type, $\mathsf{P}$ a process. $\projf$ denotes projection 
and $\rightarrow$ denotes reduction. The dotted lines correspond to relations inherited from \cite{srpaper}
while the solid lines denote relations that are new, or rewritten, in
this paper.
\end{minipage}  
\caption{Design overview}
\end{figure}
Our Rocq implementation builds upon the recent formalisation of subject reduction for MPST by 
Ekici et al.\ \cite{srpaper},
which itself is based on \cite{SynchronousSubtyping}.
The methodology in \cite{srpaper} takes an equirecursive approach where an inductive syntactic global or local type is identified 
with the coinductive tree obtained by fully unfolding the recursion. 
It then defines a coinductive projection relation between global and local type trees,
the LTS semantics for global type trees, and typing rules for the session calculus
outlined in \cite{SynchronousSubtyping}. 
We extensively use these definitions and the lemmas concerning them, but we depart from and extend 
\cite{srpaper} in numerous ways by introducing {\ltc}s, their correspondence with global types and 
a new typing relation. Our addition to the code amounts to around 14K lines of Rocq, which we link
throughout with the symbol \includegraphics[width=11pt]{img/icon-rocq-blue.png}.
\\ \indent As with \cite{srpaper}, our implementation heavily uses the parameterised coinduction technique 
of the Paco \cite{paco} library. Namely, our liveness property is defined using 
possibly infinite \textit{execution traces} which we represent as coinductive streams.
The relevant predicates on these traces, such as fairness, are then defined as mixed inductive-coinductive 
predicates using linear temporal logic (LTL) \cite{pnueli1977temporal}.
\revised{In Paco coinductive predicates are defined as the greatest fixpoints of
inductive relations parameterised by an auxiliary relation representing the 
knowledge accumulated during the course of the proof. 
Statements on these predicates can then be proved by incrementally adding to
this knowledge over the course of the proof, sidestepping the slow and 
non-compositional syntactic guardedness checks needed for the soundness of 
Rocq's native \lstin{cofix} tactic. 
Our approach using LTL and parameterised coinduction
results in compositional and clear proofs.}
\\ \indent \textbf{Outline.} In \cref{sec-procs} we define our session calculus and its LTS semantics. 
In \cref{sec:types} we recapitulate the definitions of local and global type trees, and the subtyping and projection relations on them, from \cite{srpaper}. 
In \cref{sec:lts} we give LTS
semantics to {\ltc}s and global types, and detail the association relation between them.
In \cref{sec-props} we define safety and liveness for {\ltc}s, and prove that they hold
for contexts associated with a global type tree. In \cref{sec-proc-props} we give the typing rules for 
our session calculus, and prove the desired properties of typable sessions.

\section{Synchronous Multiparty Session Calculus}
\label{sec-procs}
We introduce a simple synchronous multiparty session calculus \cite{VeryGentle} that our type system will be used on. 
\subsection{Processes and Multiparty Sessions}
\begin{definition}[Expressions and Processes]
  \label[definition]{def:processes}
  We define processes as follows:\\[1mm]
  \centerline{
$\PT ::= \prt{p}!\ell(\kf{e}).\PT \SEP \sum_{i \in I}
\prt{p}?\ell_i(x_i).\PT_i \SEP \cond{\kf{e}}{\PT}{\PT} \SEP \mu \Xv.\PT \SEP \Xv \SEP \textbf{0}$}\\[1mm]
    where \kf{e} is an expression, which is either a variable, a value such as \texttt{true}, $0$ or $-3$, 
    or a term built from expressions with operators
    such as \texttt{succ}, \texttt{neg}, $\neg$ and   
    non-deterministic choice $\sendsign$. 
\end{definition}
Process $\prt{p}!\ell(\kf{e}).\PT$ sends the value of expression
$\kf{e}$ 
with label $\ell$ to participant $\pp$, and continues with process $\PT$.
Process $\sum_{i \in I} \prt{p}?\ell_i(x_i).\PT_i$ receives a value from $\pp$ with any label  
$\ell_i$ where $i{\in} I$, with $I$ being a finite non-empty index set, binding the result to $x_i$ and continuing with $\PT_i$, depending on which $\ell_i$
the value was received with. $\Xv$ is a recursion variable, $\mu \Xv.\PT$ is a recursive process,
$\cond{\kf{e}}{\PT}{\PT}$ is a conditional and $\textbf{0}$ is a terminated process. 
\revised{We always assume that recursion is guarded}.

Processes can be composed in parallel into sessions.
\begin{definition}[Multiparty Sessions]\label[definition]{def:sessions}
  Multiparty sessions are defined as follows.
    \[
    \M \;::=\;%
    {\prt{p}} \triangleleft \PT \quad \SEP \quad(\M  \mid \M)
    \quad \SEP \quad \mathcal{O}
  \]
\end{definition}
${\prt{p}} \triangleleft \PT$ denotes that participant $\pp$ is running the process $\PT$,
$\mid$ indicates parallel composition. We write $\displaystyle  \prod_{i \in I}^{} \pp_i \triangleleft \PT_i$
to denote the session formed by $\pp_i$ running $\PT_i$ in parallel for all $i \in I$. 
$\mathcal{O}$
is an empty session with no participants, that is, the unit of parallel composition.
In Rocq processes and sessions are defined with the inductive types \lstin{process} \rocqlink{\rurl{\stbsrc process.v}{L11-L17}} and 
\lstin{session} \rocqlink{\rurl{\stlsrc session.v}{L12-L15}}. 
\revised{We follow \cite{srpaper} and represent the continuations in a receiving process using a \lstin{list} of \lstin{option} types. 
In a continuation \lstin{gcs : list (option process)}, index \lstin{k} (using zero-indexing) being equal to 
\lstin{Some P_k} means that $\ell_k (x).\PT_k$ is available in the continuation.
Similarly index \lstin{k} being equal to \lstin{None} or being out of bounds of the list means 
that the message label $\ell_k$ is not present in the continuation. The function \lstin{onth} \rocqurl{STBase/src/header}{L95-L106} 
formalises this convention in Rocq. This representation is also used for local (\cref{def:local_type_trees}) 
and global (\cref{def:global_type_trees}) type trees.}
\begin{tcbbr}{Rocq}
Inductive process : Type := 
  | p_send : part -> label -> expr -> process -> process
  | p_recv : part -> list(option process) -> process 
  | p_ite : expr -> process -> process -> process
  | p_rec : process -> process
  | p_var : nat -> process
  | p_inact : process.  
Inductive session: Type :=
  | s_ind : part   -> process -> session
  | s_par : session -> session -> session
  | s_zero : session.
\end{tcbbr}

\subsection{Structural Congruence and Operational Semantics}

We define the operational semantics for sessions by the means of a labelled transition system.
We omit the semantics for expressions as they are standard and are found in \cite{SynchronousSubtyping}. 

We start by defining a structural congruence relation $\equiv$ on sessions which expresses the commutativity, associativity and 
unit of the parallel composition operator \rocqlink{\rurl{STLive/src/session.v}{L37-L46}}.
For reductions, we use labelled \textit{reactive} semantics \cite{fairnesslock,castellani_reversible_2019}
which doesn't contain explicit silent $\tau$ actions for internal reductions (that is, evaluation of $\kf{if}$-expressions
and unfolding of recursion) while still considering $\beta$-reductions
up to those internal reductions by using an unfolding relation. 
This stands in contrast to the more standard semantics used in \cite{srpaper,SynchronousSubtyping,fairnesslock}.
For the advantages of our approach see \cref{remark-reactive-justif}.
\label{def-sess-semantics}

\begin{table}[ht]
{\footnotesize
\[
\begin{array}{@{}l@{}}
  \inferrule[\rulename{R-comm}]
  {j\in I \quad e \downarrow v}
  {\pp \triangleleft \sum^{}_{i\in I} \tin\pq{\ell_i}{x_i}.\PT_i \ \mid \  \pq \triangleleft \tout{\pp}{\ell_j}{\kf{e}}.\QT \ \mid \ \N \ \ 
  \xrightarrow{(\pp,\pq)\ell_j} \ \ 
    \pp \triangleleft \PT_j[v/x_j] \ \mid \ \pq \triangleleft \QT \ \mid \ \N}
\quad
    \inferrule[\rulename{Unf-trans}]
  {\M \Rrightarrow \M' \quad \M' \Rrightarrow \N }
  {\M \Rrightarrow \N}  
    \\[10mm]
    \inferrule[\rulename{R-unfold}]
  {\M \Rrightarrow \M' \quad \M' \lts{\lambda} \N' \quad \N' \Rrightarrow \N}
  {\M \lts{\lambda} \N}
  \qquad
    \inferrule[\rulename{Unf-struct}]
  {\M \equiv \N}
  {\M \Rrightarrow \N}
  \qquad
    \inferrule[\rulename{Unf-condt}]
  {e \downarrow \text{true}}
  {\pp \triangleleft \text{ if } e \text{ then } \PT \text{ else } \QT \ \mid \ \N \ \ 
  \Rrightarrow \ \
    \pp \triangleleft \PT \ \mid \ \N}    
      \\[5mm]
      
  \inferrule[\rulename{Unf-rec}]
  {}
  {\pp \triangleleft \mu \Xv.\PT \mid \N \ \ 
  \Rrightarrow
    \pp \triangleleft \PT[\mu \Xv.\PT / \Xv] \ \mid \ \N}  
  \quad
  \inferrule[\rulename{Unf-condf}]
  {e \downarrow \text{false}}
  {\pp \triangleleft \text{ if } e \text{ then } \PT \text{ else } \QT \ \mid \ \N \ \ 
  \Rrightarrow \ \
  \pp \triangleleft \QT \ \mid \ \N}\\[5mm]
    \inferrule[\rulename{sc-sym}]
  {}
  {\M  \mid \N  \equiv \N \mid \M} 
  \qquad
  \inferrule[\rulename{sc-assoc}]
  {}
  {(\mathcal{L}  \mid \M) \mid \N  
  \equiv \mathcal{L}  \mid (\M \mid \N)} 
  \qquad
  \inferrule[\rulename{sc-o}]
  {}
  {\M  \mid \mathcal{O}  
  \equiv \M }
\end{array}
\]}
\caption{Structural Congruence, Unfolding and Reductions of Sessions}
\label{tbl:unf}
\end{table}

In \cref{tbl:unf}, $\M \Rrightarrow \N$ means that 
$\M$ can transition to $\N$ through some internal actions, that is, a reduction that doesn't involve a communication. 
We say that $\M$ \textit{unfolds} to $\N$. Then
\rulename{R-comm} captures communications between processes, and \rulename{R-unfold} lets us consider reductions up to
unfoldings.

In Rocq, the unfolding is captured by the predicate \lstin{unfoldP : session -> session -> Prop} \rocqurl{STLive/src/session}{L54-L63}
and \lstin{betaP_lbl M lambda M'} \rocqurl{STLive/src/session}{L640-L645} denotes 
$\M \lts{\lambda} \M'$. We write $\M \lts{} \M'$ if $\M \lts{\lambda} \M'$
for some $\lambda$, which is written \lstin{betaP M M'} in Rocq.
We write $\lts{}^*$ to denote the reflexive transitive closure of
$\lts{}$, which is called \lstin{betaRtc} \rocqurl{STLive/lemma/live_proc}{L25} in Rocq. 

\section{The Type System}
We briefly recap the core definitions of local and global type trees, 
subtyping and projection from
\cite{SynchronousSubtyping}. We take an equirecursive approach and 
work directly on the possibly infinite local and global type trees 
obtained by unfolding the recursion
in guarded syntactic types; details of this approach
can be found in \cite{srpaper} and hence are omitted here.  
\label{sec:types}
\subsection{Local Type Trees}
We start by defining the sorts that will be used to type expressions, 
and local types that will be used to type single processes.
\begin{definition}[Sorts and Local Type Trees]
\label[definition]{def:local_type_trees}
We define three sorts: $\tint$, $\tbool$ and $\tnat$.
Local type trees are then defined coinductively with the following syntax \rocqurl{STBase/src/local}{L11-L14}: 
\begin{align*}
    \T ::= \quad &\tend 
     \SEP \procinset{\prt{p}}{\ell_i(\S_i)}{\T_i}{i \in I} 
     \SEP \procoutset{\prt{p}}{\ell_i(\S_i)}{\T_i}{i \in I}
\end{align*}  
\end{definition}

In the above definition, $\tend$ represents a role that has finished communicating.
 $\ltrec{\prt{p}}{\ell_i(\S_i).\T_i}_{i \in I}$ denotes a role that may, from any $i \in I$,
 with $I$ being a non-empty finite indexing set, 
receive a value of sort $S_i$ with message label $\ell_i$ and continue with $\T_i$.
Similarly, $\ltsend{\prt{p}}{\ell_i(\S_i).\T}_{i \in I}$  represents a role that may choose
to send a value of sort $S_i$ with message label $\ell_i$ and continue with $\T_i$ for any $i \in I$.
Local type trees are expressed in Rocq with the following:
\begin{tcbbr}{Rocq}
Inductive sort: Type := | sbool: sort | sint : sort | snat : sort.
CoInductive ltt: Type :=
  | ltt_end : ltt
  | ltt_recv: part -> list (option(sort*ltt)) -> ltt
  | ltt_send: part -> list (option(sort*ltt)) -> ltt.
\end{tcbbr}  
\revised{As with processes, we represent the continuations using a \lstin{list} of \lstin{option} types. 
Index $k$ of the continuation being a \lstin{Some} value means that label $\ell_k$ exists in the continuation.}
\subsection{Subtyping}
We define the subsorting relation on sorts and the process-oriented \cite{Gay2016} subtyping relation on local type trees.
\begin{definition}[Subsorting and Subtyping] \label[definition]{def:subtyping}
  Subsorting $\subso$ is the least reflexive binary relation that satisfies $\tnat \subso \tint$. 
  Subtyping $\subtp$ is the largest relation between local type trees coinductively defined by the following rules:
  \[
  \small
  \begin{array}[t]{@{}c@{}}
    \cinferrule[\rulesubend]{
    }{
    \tend \subtp \tend
    }
    \quad
  \cinferrule[\rulesubin]{
    \forall i \in I: \qquad S'_i \subso S_i \qquad \T_i \subtp \T'_i
  }
  {
    \procinset{p}{\ell_i(S_i)}{\T_i}{i \in I \cup J} \subtp 
    \procinset{p}{\ell_i(S'_i)}{\T'_i}{i \in I}
  }
  \quad 
  \cinferrule[\rulesubout]{
    \forall i \in I: \qquad S_i \subso S'_i \qquad \T_i \subtp \T'_i
  }
  {
    \procoutset{p}{\ell_i(S_i)}{\T_i}{i \in I}
    \subtp
    \procoutset{p}{\ell_i(S'_i)}{\T'_i}{i \in I \cup J}
    }
  \end{array}
  \]
\end{definition}
Intuitively, $\T_1 \subtp \T_2$ means that a role of type $\T_1$ can be supplied anywhere 
a role of type $\T_2$ is needed.
$\rulesubin$ captures the fact that we can supply a role that is able to receive 
more labels than specified, and $\rulesubout$ captures that we can supply a role
that has fewer labels available to send. Note the contravariance of the sorts in $\rulesubin$;
if the supertype demands the ability to receive an $\tnat$ then the subtype can 
receive $\tnat$ or $\tint$.

In Rocq, the subtyping relation \lstin{subtypeC : ltt -> ltt -> Prop} is expressed \rocqurl{STBase/src/local}{L196-L205} as a greatest fixpoint using the 
\lstin{Paco} library \cite{paco}; 
for details we refer to \cite{SynchronousSubtyping}.
\subsection{Global Type Trees}
We now define global types which give a bird's eye view of the whole protocol.
As before, we work directly on infinite trees and omit the details which can be found in 
\cite{srpaper}.
\begin{definition}[Global type trees]
  \label[definition]{def:global_type_trees}
  We define global type trees coinductively as follows \rocqurl{STBase/src/global}{L11-L13}:\vspace{-2pt} 
  
  \begin{minipage}{0.43\textwidth}
  \begin{align*}
\G &::= \tend 
     \SEP \GvtPair{\prt{p}}{\prt{q}}{\ell_i(\S_i).\G_i}_{i\in I}
    \end{align*}  
  \end{minipage}
  \begin{minipage}{0.5\textwidth}
    \begin{tcb}{Rocq}
CoInductive gtt: Type :=
  | gtt_end    : gtt 
  | gtt_send   : part -> part -> list (option (sort*gtt)) -> gtt.
    \end{tcb}
    \end{minipage}    
\end{definition}
$\tend$ denotes a protocol that has ended, $\GvtPair{\prt{p}}{\prt{q}}{\ell_i(\S_i).\G_i}_{i\in I}$
denotes a protocol where for any $i \in I$, with $I$ being a non-empty finite index set, participant $\pp$ may send a value of sort 
$S_i$ to another participant $\pq$ via message label $\ell_i$, after which the protocol continues as $\G_i$.
We further define a function $\funprt(\G)$ that denotes the 
participants of the global type $\G$ as the least 
solution\footnote{This is a simplified presentation of the definition of $\funprt$ in Rocq;
for technical details see \cite{srpaper}.} to the following equations:
\vspace{-2pt}
\begin{align*}
  &\funprt(\tend)=\emptyset \quad
  &\funprt(\GvtPair{\prt{p}}{\prt{q}}{\ell_i(\S_i).\G_i}_{i\in I})=
      \{\prt{p},\prt{q}\} \cup \bigcup _{i \in I} {\funprt(\G_i)}
\end{align*}
\mbox{}\\[-6pt]
In Rocq $\funprt$ is captured with the 
predicate \lstin{isgPartsC : part -> gtt -> Prop} \rocqurl{STBase/src/part}{L15-L16}, where \lstin{isgPartsC p G} 
denotes $\prt{p} \in \funprt(\G)$.  

\subsection{Projection}
We now define coinductive projections with plain merging 
(see \cite{projsurvey} for a survey of other notions of merge). 
\begin{definition}[Projection] \label[definition]{def:projections}
    The projection of a global type tree onto a participant ${\prt{r}}$ is the largest relation $\upharpoonright_{\prt{r}}$ between global type trees and
    local type trees such that, whenever $\G\proj{r}\T$:
    \begin{itemize}
      \item ${\prt{r}} \notin \participant{\G}$ implies $\T = \tend$; \hfill {\ruleprojend}
      \item $\G = \GvtPair{\prt{p}}{\prt{r}}{\ell_i(S_i).\G_i}_{i \in I}$ implies 
      $\T = \procinset{\prt{p}}{\ell_i(\S_i)}{\T_i}{i \in I}$ and $\forall i \in I, \G_i\proj{r}\T_i $ \hfill {\ruleprojin}
      \item $\G = \GvtPair{\prt{r}}{\prt{q}}{\ell_i(S_i).\G_i}_{i \in I}$ implies 
      $\T = \procoutset{\prt{q}}{\ell_i(\S_i)}{\T_i}{i \in I}$ and $\forall i \in I, \G_i\proj{r}\T_i $ \hfill {\ruleprojout}
      \item $\G = \GvtPair{\prt{p}}{\prt{q}}{\ell_i(S_i).\G_i}_{i \in I} \text{and } {\prt{r}} \notin \{{\prt{p}},{\prt{q}}\}$ implies 
      that $\forall i \in I, \G_i\proj{r}\T $ \hfill {\ruleprojcont}
    \end{itemize}
\end{definition}

Informally, the projection of a global type tree $\G$ onto a participant $\prt{r}$
extracts a role for participant  $\prt{r}$ from the protocol whose bird's-eye view 
is given by $\G$. \ruleprojend expresses that if $\prt{r}$ is not a participant of $\G$ then
$\prt{r}$ does nothing in the protocol. \ruleprojin and \ruleprojout handle the cases where $\prt{r}$
is involved in a communication in the root of $\G$. $\ruleprojcont$ says that, if $\prt{r}$ is not
involved in the root communication of $\G$ and all continuations of $\G$ project on to the same type,
then $\G$ also projects on to that type.
In Rocq, projection is defined as a Paco greatest fixpoint with the relation
\lstin{projectionC : gtt -> part -> ltt -> Prop} \rocqurl{STBase/src/projection}{L9-L32}.

Using a result from \cite{srpaper} \rocqurl{STBase/lemma/projection_helper}{L513-L530}, we can regard projection as a partial
function. We write $\projfn{\G}{\pr}=\T$ when $\G\proj{\pr}\T$. Furthermore we will frequently
be making assertions about subtypes of projections of a global type e.g.
$\T \subtp \projfn{\G}{\pr}$. In our Rocq implementation we define the predicate 
\lstin{issubProj : ltt -> gtt -> part -> Prop}
as a shorthand for this.
\subsection{Balancedness, Global Tree Contexts and Grafting}
We introduce an important constraint on the types of global type trees we will consider, \textit{balancedness}.
We omit the technical details of the definition and the Rocq implementation; they can be found in 
\cite{SynchronousSubtyping} and \cite{srpaper}.
\begin{definition}[Balanced Global Type Trees]\label[definition]{def-balance}
A path on a tree is a sequence of nodes starting from the root such that every node is a child of the preceding one. 
A global type tree $\G$ is balanced if for any subtree $\G'$ of $\G$, there exists $k$ such that
for all $\pp \in \funprt(\G')$, $\pp$ occurs on every path from the root of $\G'$ that has 
length at least $k$ or ends in $\tend$.
\end{definition}
Balancedness is a regularity condition that imposes a notion of \textit{liveness} on the protocol
described by the global type tree. Indeed, our liveness results in \cref{sec-proc-props} 
hold only for balanced global types. Another reason for formulating balancedness is that it allows 
us to use the \textit{grafting} technique, turning proofs by coinduction on infinite trees to proofs by induction
on finite \textit{global type tree contexts}, or \textit{g-contexts} for short.

\begin{definition}[g-contexts and Grafting]\label{def:global-ctx}
  g-contexts are defined inductively with the following syntax \rocqurl{STBase/src/gttreeh}{L10-L12}:
  
  \begin{minipage}{0.45\textwidth}
  \begin{align*}
    \Gcx &::= \quad \GvtPair{p}{q}{\ell_i(S_i).\Gcx_i}_{i \in I} 
   \SEP \hole_j
  \end{align*}  
  \end{minipage}
  \begin{minipage}{0.5\textwidth}
  \begin{tcb}{Rocq}
Inductive gtth: Type :=
  | gtth_hol    : fin -> gtth
  | gtth_send   : part -> part -> list (option (sort * gtth)) -> gtth.
  \end{tcb} 
  \end{minipage}
   Given a g-context $\Gcx$ whose holes are in the indexing set $J$
  and a set of global types $\{\G_j\}_{j \in J}$,  the grafting $\Gcx[\G_j]_{j \in J}$ denotes 
  the global type tree obtained by substituting $\hole_j$ with $\G_j$ in $\Gcx$.

  In Rocq the indexed set $\{\G_j\}_{j \in J}$ is represented using a \lstin{list (option gtt)}.
  Grafting is expressed with the  inductive relation 
  \lstin{typ_gtth : list (option gtt) -> gtth -> gtt -> Prop} \rocqurl{STBase/src/gttreeh}{L33-L38}.
  \lstin{typ_gtth gs gcx gt} means that the grafting of the set of global type trees 
  \lstin{gs} onto the g-context \lstin{gcx} results in the tree \lstin{gt}.
  We additionally define $\funprt$ and \lstin{ishParts} on g-contexts analogously 
  to $\funprt$ and \lstin{isgPartsC} on trees \rocqurl{STBase/src/part}{L24-L28}. 
\end{definition}
A g-context can be thought of as the finite prefix of a global type 
tree, where holes $\hole_j$ indicate the cutoff points. g-contexts
are related to global type trees with the \textit{grafting} 
operation that fills in the holes with type trees.
The following lemma relates g-contexts to balanced
global type trees. 
\begin{lemma}[Proper Grafting Lemma, \cite{srpaper} \rocqurl{STBase/lemma/decidable}{L206-L211}]
  \label[lemma]{lem-grafting}
\revised{If $\G$ is a balanced global type tree and $\pp \in \funprt(\G)$,
then there is a g-context $\Gcx$ and 
an indexed set $\{\G_i\}_{i \in I}$ such that $\Gcx[\G_i]_{i \in I} = \G$,
$\pp \notin \funprt(\Gcx)$, and for all $i \in I$, $\G_i$ is in the shape $\tend$,
$\GvtPair{p}{q}{\dots}$ or $\GvtPair{q}{p}{\dots}$. In this case, we refer to $\Gcx$ and $\{\G_i\}_{i \in I }$
as the $\pp$-grafting of $\G$. When we do not care about the $\G_i$ we may just say that $\G$ is 
$\pp$-grafted by $\Gcx$.}
\end{lemma}
\begin{tcolorbox}[colback=blue!10]
\cref{lem-grafting} allows us to turn proofs by coinduction on infinite trees 
to proofs by induction on the grafting g-context---one of the
main proof techniques used in this work.
\end{tcolorbox}
\begin{remark}
  From now on, all the global type trees we will be referring to are assumed to be balanced.
  When talking about the Rocq implementation, any \lstin{G : gtt} we mention is assumed
  to satisfy the predicate \lstin{wfgC G} \rocqurl{STBase/src/balanced}{L147}, expressing that \lstin{G} corresponds to
  a well-formed \cite[after Definition 24]{srpaper}, balanced type. 
  Furthermore, we will often require that a global type is projectable onto all its participants. 
  This is captured by
  the predicate \lstin{projectableA G = forall p, exists T, projectionC G p T}. As with \lstin{wfgC}, 
  we will be assuming that all types we mention are projectable. 
\end{remark}

\section{Semantics of Global and Local Types}
\label{sec:lts}
In this section we introduce {\ltc}s, and define Labelled Transition System 
semantics on these constructs.

\subsection{{\LTC}s and Reductions} 
We start by defining {\ltgc}s, 
also called \textit{local type contexts} in the related work \cite{LessIsMoreRevisited,PY2025,YHK2026}.

\begin{definition}[{\LTgC}s]\label[definition]{def-type-ctx}
{\Ltc}s are defined as a finite mapping of participants to local type trees
with the following \rocqurl{STLive/src/lcontext}{L22-L26}:

\begin{minipage}{0.4\textwidth}
\begin{align*}
    \Gamma \;::=\;%
    \emptyset \SEP \Gamma, {\prt{p}}:\T
  \end{align*}
\end{minipage}
\begin{minipage}{0.5\textwidth}
  \begin{tcb}{Rocq}
Module M := MMaps.RBT.Make(Nat).
Module MF := 
  MMaps.Facts.Properties Nat M.
Definition tctx: Type := M.t ltt.
\end{tcb}
\end{minipage}
\end{definition}
Intuitively, $\pp$ : {\T} means that participant p is associated with a process that has the type tree T. 
We write $\dom{\Gamma}$ to denote the set of participants occurring in $\Gamma$, and 
abbreviate the singleton type {\ctext} consisting of the pair $\pp$ and  
{\T} as $\pp : \T$. We write $\Gamma(\pp)$ for the type of $\pp$ in $\Gamma$. 
We define the composition $\Gamma_1,\Gamma_2$ iff $\dom{\Gamma_1} \cap \dom{\Gamma_2}=\emptyset$.

In the Rocq implementation we implement {\ltgc}s \revised{\lstin{tctx}} as finite maps of 
participants, which are represented as natural numbers, and local type trees.
We use the finite map implementation of the MMaps library \cite{mmaps}.\pagebreak[3]
We further enforce the non-emptiness of the indexing sets of the local type trees
by positing that for any local type tree in any {\ctext} we mention, 
the well-formedness predicate \lstin{wflttC T} \rocqurl{STLive/src/wfltt}{L10-L18} holds. This is expressed by the predicate
\lstin{tctx_wf: tctx -> Prop} \rocqurl{STLive/src/wfltt}{L55-L56}.

We give LTS semantics to {\ltgc}s.
\begin{definition}[Transition labels] A transition label $\alpha$ has the following form:
  \begin{align*}
    \alpha ::&= \lblrec{p}{q}{\ell(S)}  && 
    \text{(${\prt{p}}$ receives a value of sort ${S}$ from ${\prt{q}}$ with message label ${\ell}$)}\\
    & \SEP \lblsend{p}{q}{\ell(S)} && \text{(${\prt{p}}$ sends a value of sort $S$ to ${\prt{q}}$ with message label $\ell$ )}\\
    & \SEP \lblsync{p}{q}{\ell} && \text{(A synchronised communication from $\pp$ to $\pq$ occurs via label $\ell$)}
  \end{align*}
  We further define the function $\subject{\alpha}$ as \enspace 
  $\subject{\lblrec{p}{q}{\ell(S)}}=\subject{\lblsend{p}{q}{\ell(S)}}=\{\pp\}$ and $\subject{\lblsync{p}{q}{\ell}}=\{\pp,\prt{q} \}$.
\end{definition}
\begin{definition}[Typing {\ctext} reductions] \label[definition]{def-ctx-red}
The typing {\ctext} transition $\lts{\alpha}$ is defined inductively by the following rules:
  \[
\begin{array}[t]{@{}c@{}}
\inferrule[]{
     k \in I }{
      \lbltrans
    {\prt{p} : \procinset{\prt{q}}{\ell_i(\S_i)}{\T_i}{i \in I}} 
    {\lblrec{p}{q}{\ell_k(S_k)}}
    {\prt{p} : \T_k}
    }
    \; \rulename{$\Gamma$-\andsign}
    \;
  \inferrule[]{
     k \in I }{
      \lbltrans
      {\prt{p} : \procoutset{\prt{q}}{\ell_i(\S_i)}{\T_i}{i \in I}} 
      {\lblsend{p}{q}{\ell_k(S_k)}}
      {\prt{p} : \T_k}
    }
    \; \rulename{$\Gamma$-$\oplus$}
  \\\\
    \inferrule[]{
  \lbltrans{\Gamma}{\alpha}{\Gamma'}}
  {
    \lbltrans{\Gamma, \prt{p} : \T}
     {\alpha} {\Gamma', {\prt{p}} : \T} 
  }
  \enspace \rulename{$\Gamma$-,}
  
  \qquad
    \inferrule[]{
    \lbltrans{\Gamma_1}{\lblsend{\prt{p}}{\prt{q}}{\ell(S)}}{\Gamma'_1}
    \qquad  
    \lbltrans{\Gamma_2}{\lblrec{\prt{q}}{\prt{p}}{\ell(S')}}{\Gamma'_2}
    \qquad S \subso \S'
    }{
    \lbltrans{\Gamma_1, \Gamma_2 }
    {\lblsync{p}{q}{\ell}}
    {\Gamma'_1, \Gamma'_2}
    }
    \enspace \rulename{$\Gamma$-$\oplus$\andsign}
\end{array}
\]
We write $\Gamma \lts{\alpha}$ if there exists $\Gamma'$ such that $\Gamma\lts{a}\Gamma'$.  
We define a reduction $\Gamma\lts{}\Gamma'$ to hold iff \; $\lbltrans{\Gamma}{\lblsync{\prt{p}}{\prt{q}}{\ell}}{\Gamma'}$ for some ${\prt{p}}$, ${\prt{q}}$, $\ell$. We write $\Gamma\lts{}$ iff \; $\Gamma\lts{}\Gamma'$ for some $\Gamma'$. 
We write $\lts{}^*$ for the reflexive transitive closure of $\lts{}$.
\end{definition}
$\ruleredsend$ and $\ruleredrec$, express a single participant sending or receiving.
$\ruleredsync$ expresses a synchronised communication where one participant sends while another receives,
and they both progress with their continuation. $\ruleredvar$ shows how to extend an {\ctext}.  
In Rocq typing {\ctext} reductions are defined with the predicate \lstin{tctxR} \rocqurl{STLive/src/lcontext}{L45-L68}.

\begin{tcbbr}{Rocq}
Inductive tctxR: tctx -> label -> tctx -> Prop :=
  | Rsend: ...
  | Rrecv: ...  
  | Rcomm: ...
  | RvarI: ...
  | Rstruct: forall g1 g1' g2 g2' l, tctxR g1' l g2' ->
      M.Equal g1 g1' -> M.Equal g2 g2' -> tctxR g1 l g2.
\end{tcbbr}

The first four constructors in the definition of \lstin{tctxR} correspond to the 
rules in \cref{def-ctx-red}, and \lstin{Rstruct} expresses the indistinguishability
of local {\ctext}s under the \lstin{M.Equal} predicate from the MMaps library. 
\lstin{M.Equal} lets us consider two finite maps with the same keys mapping to the same values as equal,
and is the main notion of equality we use for typing {\ctext}s in this paper.

We illustrate typing {\ctext} reductions with an example.
\begin{example}\label{exam:reductions}
  Let $\Gamma = \{{\prt{p}}:\T_{\prt{p}}, \; {\prt{q}} : \T_{\prt{q}},\; {\prt{r}}: \T_{\prt{r}}\}$ where
  $\T_{\prt{p}} =
  \procoutmult{\prt{q}}{\ell_0(\tint).\T_{\prt{p}},\ell_1(\tint).\tend}$, 
  $\T_{\prt{q}} = \procinmult{\prt{p}}{\ell_0(\tint).\T_{\prt{q}},  \ell_1(\tint).\procoutsingle{\prt{r}}{\ell_2(\tint)}{\tend}}$
  and 
  $\T_{\prt{r}} = \procinmult{\prt{q}}{\ell_2(\tint).\tend}$.
We have the reductions $\Gamma\lts{\lblsend{p}{q}{\ell_0(\tint)}} \Gamma$ and
$\Gamma \lts{\lblrec{\prt{q}}{\prt{p}}{\ell_0(\tint)}} \Gamma$, which synchronise to give the reduction
and $\Gamma \lts{\lblsync{\prt{p}}{\prt{q}}{\ell_0}}  \Gamma$. Similarly via synchronised communication of $\pp$ and $\pq$
via message label $\ell_1$ we get $\Gamma \lts{\lblsync{\prt{p}}{\prt{q}}{\ell_1}} \Gamma'$
where $\Gamma'$ is defined as  $\{{\prt{p}} : \tend, \; {\prt{q}}: 
\procoutsingle{\prt{r}}{\ell_2(\tint)}{\tend}, {\prt{r}} : \T_{\prt{r}}\}$.
We further have that $\Gamma' \lts{\lblsync{q}{r}{\ell_2}} \Gamma_\mathtt{end}$ where
$\Gamma_{\mathtt{end}}$ is defined as 
$\{{\prt{p}}:\tend, \; {\prt{q}} : \tend,\; {\prt{r}}: \tend \}$.

In Rocq, $\Gamma$ is defined the following way \rocqurl{STLive/examples/ex_type_semantics}{L12-L21}:

\begin{tcbbr}{Rocq}
Definition prt_p:=0.
Definition prt_q:=1.
Definition prt_r:=2.
CoFixpoint T_p := ltt_send prt_q [Some (sint,T_p); Some (sint,ltt_end); None].
CoFixpoint T_q := ltt_recv prt_p [Some (sint,T_q); Some (sint, ltt_send prt_r [None;None;Some (sint,ltt_end)]); None].
Definition T_r := ltt_recv prt_q [None;None; Some (sint,ltt_end)].
Definition gamma := M.add prt_p T_p (M.add prt_q T_q (M.add prt_r T_r M.empty)).
\end{tcbbr}

Now $\Gamma \lts{\lblsync{\prt{p}}{\prt{q}}{\ell_0}} \Gamma$  can be expressed as 
\lstin{tctxR gamma (lcomm prt_p prt_q 0) gamma} \rocqurl{STLive/examples/ex_type_semantics}{L63}.
\end{example}

\subsection{Global Type Reductions}
As with {\ltgc}s, we can also define reductions for global types.
\begin{definition}[Global type reductions]
  The global type transition $\lts{\alpha}$ is defined coinductively as follows.
  \[
  \small
\begin{array}[t]{@{}c@{}}
\cinferrule[]{
     k \in I }{
      \lbltrans
    {\GvtPair{p}{q}{\ell_i(S_i).\G_i}_{i \in I}} 
    {\lblsync{p}{q}{\ell_k}}
    {\G_k}
    }
    \enspace \rulename{GR-$ \sendsign \andsign$}
    \\\\
    \cinferrule[]{
     \forall i \in I \enspace \lbltrans{\G_i}{\alpha}{\G'_i} \qquad
     \subject{\alpha} \cap \{\prt{p},\prt{q}\} = \emptyset
     \qquad \forall i \in I \enspace \{\prt{p},\prt{q}\} \subseteq \funprt(\G_i)
     }{
      \lbltrans
    {{\GvtPair{p}{q}{\ell_i(S_i).\G_i}_{i \in I}}} 
    {\alpha}
    {{\GvtPair{p}{q}{\ell_i(S_i).\G'_i}_{i \in I}}}
    }
    \enspace \rulename{GR-Ctx}
\end{array}
\]
\end{definition}
\rulegredsendrec says that a global type tree with root $\pp \rightarrow \pq$ can
transition to any of its children corresponding to the message label chosen by $\pp$. 
\rulegredctx says that if the subjects of $\alpha$ are disjoint from the root and all its children
can transition via $\alpha$, then the whole tree can also transition via $\alpha$, with the root remaining 
the same and just the subtrees of its children transitioning.
In Rocq global type reductions are expressed using the coinductively defined predicate \lstin{gttstepC} \rocqurl{STBase/src/step}{L9-L25}. 
For example, $\lbltrans{\G}{\lblsync{p}{q}{\ell_k}}{\G'}$ translates to \lstin{gttstepC G G' p q k}. 
We refer to \cite{srpaper} for details.

\subsection{Association Between {\LTC}s and Global Types}
We have defined {\ltc}s, which specify protocols bottom-up 
by directly describing the roles of every participant,
and global types, which give a top-down view of the whole protocol, and the transition relations on them.
We relate these local and global definitions by defining \textit{association} between local type
{\ctext} and global types.
\begin{definition}[Association]\label[definition]{def-assoc}
A {\ltgc} $\Gamma$ is associated with a global type tree $\G$, written $\Gamma \assoc \G$,
if the following hold:
\begin{itemize}
  \item For all $\pp \in \funprt(\G)$, $\pp \in \dom{\Gamma}$ and $\Gamma(\pp) \subtp \G \projf{\pp}$.
  \item For all $\pp \notin \funprt(\G)$, either $\pp \notin \dom{\Gamma}$ or $\Gamma(\pp)=\tend$. 
\end{itemize}
\revised{In Rocq this is expressed with the predicate \lstin{assoc : tctx -> gtt -> Prop} \rocqurl{STLive/src/assoc}{L20-L23}}.
\end{definition} 
Informally, $\Gamma \assoc \G$ says that the local type trees in $\Gamma$ 
obey the specification described by the global type tree $\G$. 
\begin{example}\label{exam:assoc}
  In Example \ref{exam:reductions},
  we have that $\Gamma \assoc \G$ where
  $\G := \GvtPair{\prt{p}}{\prt{q}}{\ell_0(\tint).\G,\ell_1(\tint).\GvtPair{\prt{q}}{\prt{r}}{\ell_2(\tint).\tend}}$
  (note that $\G$ is not a balanced global type tree due to the
  infinite path of $\ell_0$ communications that do not involve $\pr$).
  In fact, we have $\Gamma(\ps) = \projfn{\G}{\ps}$ for $\ps \in \{\prt{p},\pq,\pr\}$. 
  Similarly, we have $\Gamma' \assoc \G'$ where
    $\G' := \GvtPair{\prt{q}}{\prt{r}}{\ell_2(\tint).\tend}$.
\end{example}

It is desirable to have the association
be preserved under {\ltc} and global type reductions, that is, 
when one of the associated constructs "takes a step" so should the other. We formalise this \textit{operational correspondence} 
property with 
the following soundness and completeness theorems.

\begin{theorem}[Soundness of Association \rocqurl{STLive/lemma/soundness}{L1075-L1080}]\label{theo-soundness} 
  \revised{If $\Gamma \assoc \G$ and $\G \lts{(\pp,\pq)\ell} \G'$, then there is a {\ltc} $\Gamma'$,
  a global type $\G''$ and a message label $\ell'$ such that $\G \lts{(\pp,\pq)\ell'} \G''$, $\Gamma' \assoc \G''$
  and $\Gamma \lts{(\pp,\pq)\ell'} \Gamma'$.}
\end{theorem}
\begin{theorem}[Completeness of Association \rocqurl{STLive/lemma/completeness}{L455-L459}] \label{theo-completeness}
  \revised{If $\Gamma \assoc \G$ and $\Gamma \lts{(\pp,\pq)\ell} \Gamma'$,
  then there exists a global type tree $\G'$ such that $\Gamma' \assoc \G'$ and 
  $\G \lts{(\pp,\pq)\ell} \G'$}. 
\end{theorem}
\begin{remark}
  Note that in the statement of soundness we allow the message label for the {\ltc} reduction 
  to be different from the message label for the global type reduction. 
  This is because our use of subtyping in association causes the entries in the {\ltc}
  to be less expressive than the types obtained by projecting the global type. For example consider
    $\Gamma = \pp : \ltsend{q}{\ell_0(\tint).\tend}, \; \pq : \ltrec{p}{\ell_0(\tint).\tend, \ell_1(\tint).\tend}$
    and $\G= \GvtPair{p}{q}{\ell_0(\tint).\tend, \ell_1(\tint).\tend}$.
  We have $\Gamma \assoc \G$ and $\G \lts{\lblsync{p}{q}{\ell_1}}$.
  However $\Gamma \lts{\lblsync{p}{q}{\ell_1}}$ is not a valid transition.
  \revised{But the $\Gamma \lts{\lblsync{p}{q}{\ell_0}}$ transition that involves the same participants is valid, so the 
  {\ltc} can still match the transition of the global type to some extent.} 
\end{remark}

\section{Properties of {\LTC}s} \label{sec-props}
We now use the LTS semantics to define some desirable properties of type {\ctext}s and their 
reduction sequences. Namely, we formulate
safety, fairness and liveness properties based on the definitions in \cite{LessIsMoreRevisited}.\footnote{Whereas in general, "safety" and "liveness" refer to classes of properties, following \cite{LessIsMore} we here refer to specific safety and liveness properties, ones that are particularly relevant for MPST.} 
\subsection{Safety}
We start by defining the \textit{safety} property that plays an important role 
in bottom-up session type systems \cite{LessIsMore}:
\begin{definition}[Safe Local Type {\Ctext}s] \label[definition]{def:safety}
  We define $\safe$ coinductively as the largest set of local type {\Ctext}s such that whenever we have $\Gamma \in \safe$:
  \begin{align}
    \label{saferule:sendrec} \tag*{\rulesafesync} & \quad \lbltrans{\Gamma}{\lblsend{p}{q}{\ell(S)}}{} \; and \; \lbltrans{\Gamma}{\lblrec{q}{p}{\ell'(S')}} \; implies 
    \; \lbltrans{\Gamma}{\lblsync{p}{q}{\ell}}\\
    \label{saferule:reduce} \tag*{\rulesafereduce} & \quad  \lbltrans{\Gamma}{}{\Gamma'} \;implies\; \Gamma' \in \safe
  \end{align}
  We write $\safe(\Gamma)$ if \;$\Gamma \in \safe$.
\end{definition}
\begin{tcolorbox}[colback=blue!10]
Safety says that if $\pp$ and $\pq$ attempt to communicate with each other and $\pp$ requests to 
send a value using message label $\ell$, then $\pq$ should be able to receive that message label.
Furthermore, this property should be preserved under any type {\ctext} reductions. 
\end{tcolorbox}

Being a coinductive property, to show that $\safe(\Gamma)$, it suffices to give a set
$\varphi$ such that $\Gamma \in \varphi$ and $\varphi$ satisfies
\rulesafesync and \rulesafereduce.
This amounts to showing that every element of $\Gamma'$ of the set of reducts of $\Gamma$,
defined
$\varphi := \{\Gamma' \; | \; \Gamma \lts{}^* \Gamma' \}$, satisfies \rulesafesync. 
We illustrate this with some examples:
\begin{example}\label{exam:safe}
   Let $\Gamma=\prt{p}:\procoutsingle{q}{\ell_0(\tint)}{\tend},\prt{q}:\procinsingle{p}{\ell_0(\tnat)}{\tend} $. 
   $\Gamma$ is not safe \rocqurl{STLive/examples/ex_type_semantics}{L532} as we have $\lbltrans{\Gamma}{\lblsend{p}{q}{\ell_0}}{}$ and 
   $\lbltrans{\Gamma}{\lblrec{q}{p}{\ell_0}}{}$ 
   but we do not have $\lbltrans{\Gamma}{\lblsync{p}{q}{\ell_0}}$ as $\tint \nleqslant \tnat $.
   
  Consider $\Gamma$ from Example \ref{exam:reductions}.  All the reducts satisfy \rulesafesync, hence $\Gamma$ is safe \rocqurl{STLive/examples/ex_type_semantics}{L370}.
\end{example}
\vspace{-5pt}
In Rocq, we define $\safe$ coinductively with Paco \rocqurl{STLive/src/path_props}{L144-L149}:

\begin{tcbbr}{Rocq}
Definition weak_safety (c: tctx ) :=
  forall p q s s'  k k', tctxRE (lsend p q (Some s) k) c -> 
  tctxRE (lrecv q p (Some s') k') c -> tctxRE (lcomm p q k) c.
Inductive safe (R: tctx -> Prop): tctx -> Prop :=
  | safety_red :  forall c, weak_safety c -> 
  (forall p q c' k,  tctxR c (lcomm p q k) c' -> exists c'', M.Equal c' c'' /\ R c'') 
    ->  safe R c.
Definition safeC c := paco1 safe bot1 c.
\end{tcbbr} 

In the above, \lstin{weak_safety} corresponds to \!\rulesafesync where \lstin{tctxRE l c} is shorthand for
\lstin{exists c', tctxR c l c'}. 
In the type \lstin{safe}, the constructor \lstin{safety_red} corresponds to 
\rulesafereduce (up to the predicate \lstin{M.Equal}). 
Then \lstin{safeC} is defined as the greatest fixed point of \lstin{safe} 
\revised{and
constructed by the Paco function \lstin{paco1 safe bot1 c}}.

We have that {\ltc}s with associated global types are always safe.
\begin{theorem}[Safety by Association \rocqurl{STLive/lemma/safety}{L44-L46}]
  \revised{If $\Gamma \assoc \G$ then $\safe(\Gamma)$. }
\end{theorem}
\subsection{Fairness and Liveness}
\label{subsec-ltl}
We now focus our attention on fairness and liveness. 
We first restate the definition of fairness and liveness for 
{\ltc} paths from \cite{LessIsMoreRevisited}.
\begin{definition}[Fair, Live Paths] \label[definition]{def-fair-live}
A {\ltc} reduction path (also called an execution or a run) is a possibly infinite sequence of transitions
$\Gamma_0 \lts{\lambda_0} \Gamma_1 \lts{\lambda_1} ..$ such that $\lambda_i$ is a synchronous transition label,
that is, of the form $\lblsync{p}{q}{\ell}$, for all $i$.

We say that a {\ltc} reduction path $\Gamma_0 \xrightarrow{\lambda_0} \Gamma_1 \xrightarrow{\lambda_2} ..$ is \emph{fair} if,
for all valid $n \mathop{\in} \mathbb{N}: 
\Gamma_n \lts{\lblsync{\prt{p}}{\prt{q}}{\ell}}$ implies 
$\exists k,\ell'$ such that $k \ge n$ and 
$\lambda_k = \lblsync{\prt{p}}{\prt{q}}{\ell'}$, and therefore
$\Gamma_k \lts{\lblsync{\prt{p}}{\prt{q}}{\ell'}} \Gamma_{k+1}$.
We say that a path $(\Gamma_n)_{n \in N}$ is \emph{live} iff, $\forall n \in N$:
\begin{enumerate}
  \item $\Gamma_n \lts{\lblsend{\prt{p}}{\prt{q}}{\ell(S)}}$ implies $\exists k,\ell'$ such that $N \ni k \ge n$ and $\Gamma_k \lts{\lblsync{\prt{p}}{\prt{q}}{\ell'}} \Gamma_{k+1}$
  \item $\Gamma_n \lts{\lblrec{\prt{q}}{\prt{p}}{\ell(S)}}$ implies $\exists k,\ell'$ such that $N \ni k \ge n$ and $\Gamma_k \lts{\lblsync{\prt{p}}{\prt{q}}{\ell'}} \Gamma_{k+1}$
\end{enumerate}
\end{definition}
\begin{definition}[Live {\LTC}]\label[definition]{def-live-ctx}
  A {\ltc} $\Gamma$ is live if whenever $\Gamma \rdc^{*} \Gamma'$,
  every fair path starting from $\Gamma'$ 
  is also live.
\end{definition}
\begin{tcolorbox}[colback=blue!10]
Informally, liveness says that every communication request on the path is eventually answered.
With our fairness assumption \cite{fairness}, we focus on "sensible" reduction paths
where every communication that's enabled by both participants is eventually executed.
Live type {\ctext}s are then defined to be the $\Gamma$ such that whenever 
$\Gamma$ can evolve (in possibly multiple steps) into $\Gamma'$, 
all fair paths that start from  $\Gamma'$ are also live.  
\end{tcolorbox}

\begin{example}\label[examples]{exam:live}
  Consider the {\ctext}s $\Gamma, \Gamma'$ and  $\Gamma_\tend$ from 
  Example \ref{exam:reductions}. One possible reduction path is 
  $\Gamma \lts{\lblsync{\prt{p}}{\prt{q}}{\ell_0}} \Gamma \lts{\lblsync{\prt{p}}{\prt{q}}{\ell_0}} \dots$. 
  Denote this path as $(\Gamma_n)_{n \in \mathbb{N}}$, where $\Gamma_n = \Gamma$ for all $n \in \mathbb{N}$. 
We have $\forall n, \Gamma_n \lts{\lblsync{\prt{p}}{\prt{q}}{\ell_0}}$ and 
  $\Gamma_n \lts{\lblsync{\prt{p}}{\prt{q}}{\ell_1}}$ as the only possible synchronised reductions from $\Gamma_n$.
  Accordingly, we also have $\forall n, \Gamma_n \lts{\lblsync{\prt{p}}{\prt{q}}{\ell_0}} \Gamma_{n+1}$ in 
  the path so this path is fair \rocqurl{STLive/examples/ex_type_semantics}{L409}.
  However, this path is not live \rocqurl{STLive/examples/ex_type_semantics}{L430} as we have $\Gamma_1 \lts{\lblrec{\prt{r}}{\prt{q}}{\ell_2(\tint)}}$ but there is no $n, 
  \ell'$ with $\Gamma_n \lts{\lblsync{\prt{q}}{\prt{r}}{\ell'}} \Gamma_{n+1}$ in the path.
  Consequently, $\Gamma$ is not a live type {\ctext}.

  Now consider the reduction path $\Gamma \lts{\lblsync{\prt{p}}{\prt{q}}{\ell_0}} \Gamma 
  \lts{\lblsync{p}{q}{\ell_1}}  \Gamma' \lts{\lblsync{\prt{q}}{\prt{r}}{\ell_2}} 
  \Gamma_\tend$. 
  This path is fair \rocqurl{STLive/examples/ex_type_semantics}{L464} and live \rocqurl{STLive/examples/ex_type_semantics}{L500} as it contains the $\lblsync{\prt{q}}{\prt{r}}{}$ transition 
  from the counterexample above.   
\end{example}
\cref{def-fair-live}, while intuitive, is not really convenient for a Rocq formalisation due to 
\revised{its explicit use of indices to quantify over the contexts in a path. 
Proofs in this setting would require additional bookkeeping to keep track of all the indices used. 
All this extra complexity could be avoided} if these properties could
be expressed as a least or greatest fixed point, which could then be formalised via Rocq's 
inductive or (via Paco) coinductive  types. 
To achieve this, we recast fairness and liveness for {\ltc} paths 
in Linear Temporal Logic (LTL) \cite{pnueli1977temporal}.  
The LTL operators \textit{eventually} ($\lozenge$) and \textit{always} ($\square$) are characterised as 
least and greatest fixed points using their expansion laws \cite[Chapter 5.14]{baier}.
\vspace{2pt} Hence they are implemented in Rocq as the inductive type \lstin{eventually} \rocqurl{STLive/src/path_props}{L74-L76} and the coinductive type 
\lstin{alwaysCG} \rocqurl{STLive/src/path_props}{L93}. We can further represent reduction paths
as \textit{cosequences}, or \textit{streams}. Then the Rocq definition of \cref{def-fair-live} 
amounts to the following \rocqurl{STLive/src/path_props}{L132-L138}:\vspace{5pt}
\begin{tcb}{Rocq}
CoInductive coseq (A: Type): Type :=
  | conil : coseq A
  | cocons: A -> coseq A -> coseq A.
Notation local_path := (coseq (tctx*option label)).
Definition fair_path_local_inner (pt: local_path): Prop :=
  forall p q n, to_path_prop (tctxRE (lcomm p q n)) False pt ->  
  eventually (headComm p q) pt.
Definition fair_path := alwaysCG fair_path_local_inner.
Definition live_path_inner (pt: local_path) : Prop := forall p q s n, 
(to_path_prop (tctxRE (lsend p q (Some s) n)) False pt -> eventually (headComm p q) pt) /\
(to_path_prop (tctxRE (lrecv p q (Some s) n)) False pt -> eventually (headComm q p) pt).
Definition live_path := alwaysCG live_path_inner.
\end{tcb}  

With these definitions we can now prove that {\ltc}s associated with a global type are live,
which is the most involved of the results mechanised in this work.
\begin{remark}
  We once again emphasise that all global types mentioned are assumed to be balanced (\cref{def-balance}).
  Indeed association with non-balanced global types doesn't guarantee liveness. 
  As an example, consider $\Gamma$ from Example \ref{exam:reductions}, which is associated with 
  $\G$ from Example \ref{exam:assoc}. Yet we have shown in Example \ref{exam:live} that $\Gamma$
  is not a live type {\ctext}. This is not surprising as $\G$ is not balanced.  
\end{remark}
\begin{theorem}[Liveness by Association \rocqurl{STLive/lemma/liveness}{L3147-L3150}]\label{theo-ctx-live}
\revised{If $\Gamma \assoc \G$ then $\Gamma$ is live.}  
\end{theorem}
\begin{proof} \revised{
(Outline) Our proof proceeds in two steps. 
First, we prove that the type {\ctext} obtained by direct projections%
\footnote{The actual Rocq proof defines an equivalent "enabledness" predicate on 
global types instead of working with direct projections. 
The outline given here is a slightly simplified presentation.} of $\G$,
that is, $\Gamma_{\texttt{proj}} = \{ \pp_i : \G \upharpoonright \pp_i \SEP \pp_i \in \funprt(\G)\}$, 
is live \rocqurl{STLive/lemma/liveness}{L3080}. 
We then leverage \cref{theo-soundness} and
\cref{theo-completeness} to show that if $\Gamma_{\texttt{proj}}$ is live, so is $\Gamma$ \rocqurl{STLive/lemma/path_assoc}{L548-L551}.
}

\revised{
Suppose  $\Gamma_{\texttt{proj}} \lts{\lblsend{p}{q}{\ell(S)}}$ (the case for the receive is similar and omitted), 
and $\rho$ is a fair {\ltc} reduction path beginning with $\Gamma_{\texttt{proj}}$.
To show that $\rho$ is live we need to show the 
existence of a  $\lblsync{p}{q}{\ell}$ transition in $\rho$. 
We achieve this by taking the height of the $\pp$-grafting of the 
global type associated with the head of $\rho$ as our induction invariant.
We show that this invariant is bounded from below by the height of the $\pq$-grafting,
and that it keeps decreasing (\rocqurl{STLive/lemma/liveness}{L610-L619}, \rocqurl{STLive/lemma/liveness}{L1765-L1778}) 
until the $\pp$-graftings and $\pq$-graftings become equal, at which point a $\lblsync{p}{q}{\ell}$ 
transition is enabled on the path \rocqurl{STLive/lemma/multigrafting}{L109-L114}.
Our fairness assumption then forces that transition to fire \rocqurl{STLive/lemma/liveness}{L2085-L2089} on $\rho$. 
}

\revised{
In the second step of the proof, we extend association on to paths \rocqurl{STLive/lemma/path_assoc}{L66-L76} to obtain, 
for each {\ltc} reduction path 
$\sigma$ that begins with $\Gamma$, another {\ltc} reduction path $\sigma'$ beginning with $\Gamma_{\texttt{proj}}$
such that the elements of $\sigma$ are subtypes (subtyping on {\ctext}s defined pointwise) 
of the corresponding elements of $\sigma'$, 
and the corresponding transitions of $\sigma$ and $\sigma'$ have the same labels. \rocqurl{STLive/lemma/path_assoc}{L500-L503}. 
This is obtained from \cref{theo-completeness}; however, in Rocq, the statement of
\cref{theo-completeness} is implemented as an \lstin{exists} statement
that lives in \lstin{Prop}, hence we need to use the \lstin{constructive_indefinite_description} axiom
to construct a \lstin{CoFixpoint} returning the desired cosequence $\sigma'$ \rocqurl{STLive/lemma/path_assoc}{L365-L385}.
Then \cref{theo-soundness} and \cref{theo-completeness} are used to show that the liveness of 
$\Gamma_{\texttt{proj}}$ implies the liveness of $\Gamma$, completing the proof.
}
\end{proof}

\section{Properties of Multiparty Sessions}\label{sec-proc-props}

We define typing rules for the session calculus introduced in \cref{sec-procs}, and prove subject 
reduction and deadlock freedom for them. Then we define 
a liveness property for sessions, and show that processes typable by a {\ltc} 
that's associated with a global type tree are guaranteed to satisfy this liveness property.
\subsection{Typing rules}
We give typing rules for our session calculus based on \cite{SynchronousSubtyping} and \cite{srpaper}. 
We have two kinds of typing judgements and type {\ctext}s.
$\Theta \vdashp \PT : \T$ says that the single process $\PT$
can be typed with local type $\T$ using expression and type variables from $\Theta$.
On the other hand, $\Gamma \vdashm \M$ expresses that session $\M$ 
can be typed by the {\ltc} (\cref{def-type-ctx}).
Typing rules for expressions are standard and can be found in e.g. \cite{SynchronousSubtyping}, and are therefore omitted.
\begin{table}[h]
{\footnotesize
\[
\begin{array}{@{}l@{}}
  \inferrule[\rulename{t-end}]
  {}
  {\Theta \vdashp \inact \colon \tend}
\quad
  \inferrule[\rulename{t-var}]
  {}
  {\Theta,\Xv\colon\T \vdashp \Xv\colon\T}
 \quad
  \inferrule[\rulename{t-rec}]
  {\Theta,\Xv\colon\T \vdashp \PT\colon \T}
  {\Theta \vdashp \mu\Xv.\PT\colon\T}
 \quad
  \inferrule[\rulename{t-if}]
  {\Theta\vdashp e\colon \texttt{bool} \quad \Theta \vdashp \PT_1\colon \T  \quad \Theta \vdashp \PT_2\colon \T}
  {\Theta \vdashp \texttt{if}\ \kf{e} \ \mathtt{then}\ \PT_1  \ \mathtt{else} \ \PT_2\colon \T}
      \\[4mm]
  \inferrule[\rulename{t-sub}]
  {\Theta \vdashp \PT\colon \T \quad \T\leqslant \T'}
  {\Theta \vdashp \PT\colon \T'}
  \quad
   \inferrule[\rulename{t-in}]
  {\forall i\in I,\quad \Theta,x_i\colon\ST_i \vdashp \PT_i\colon \T_i}
  {\Theta \vdashp \sum_{i \in I} \prt{p} ?\ell_i(x_i).\PT_i\colon  
  \ltrec{\pp}{\ell_i(\ST_i).\T_i}_{i \in I}}
 \quad
   \inferrule[\rulename{t-out}]
  {\Theta \vdashp e\colon\ST \quad \Theta \vdashp \PT\colon \T}
  {\Theta \vdashp \prt{p}!\ell(\kf{e}).\PT \ \colon
    \ {\pp}\sendsign\{{\ell(\ST).\T}\}}\\[5mm]
\inferrule[\rulename{t-sess}]{\forall i \in I: \qquad \vdashp \PT_i \colon \Gamma(\pp_i) \qquad 
\Gamma \assoc \G}
{\Gamma \vdashm \prod_i \pp_i \triangleleft \PT_i}
\end{array}
\]}
\caption{The typing rules for processes and multiparty sessions}
\label{tbl:proc}\vspace{-10pt}
\end{table}

\cref{tbl:proc} states the standard \cite{srpaper,SynchronousSubtyping} typing rules for processes, 
which we do not elaborate on.
The main rule for typing multiparty sessions is \rulename{t-sess}:  
it states that a session made of the parallel composition of processes $\prod_i \pp_i \triangleleft \PT_i$ 
can be typed by an associated local {\ctext} $\Gamma$ if the local type of participant $\pp_i$ in $\Gamma$
types the process $\PT_i$. This is expressed in Rocq with the predicate \lstin{typ_sess : session -> tctx -> Prop} \rocqurl{STLive/src/session}{L632-L638}.
\subsection{Properties of Typed Sessions}

We can now prove some properties of typed sessions. The following theorems relating
session reductions to types underlie our results.
\begin{lemma}[Typing after Unfolding \rocqurl{STLive/lemma/subj_red_helpers}{L105}]
  \label{lem-typ-unfold}
  \revised{If $\Gamma \vdashm \M$ and $\M \Rrightarrow \M'$ then $\Gamma \vdashm \M'$.} \pagebreak[3]
\end{lemma}
\begin{theorem}[Subject Reduction \rocqurl{STLive/lemma/subj_red_prog_fid}{L297}]\label{theo-sub-red}
\revised{
If $\Gamma \vdashm \M$ and $\M \lts{\lblsync{p}{q}{\ell}} \M'$, then 
there exists a type {\ctext} $\Gamma'$ such that 
$\Gamma \lts{\lblsync{p}{q}{\ell}} \Gamma'$ 
and $\Gamma' \vdashm \M'$ .
}
\end{theorem}
\begin{theorem}[Session Fidelity \rocqurl{STLive/lemma/subj_red_prog_fid}{L632-L634}]\label{theo-sess-fid}
\revised{
If $\Gamma \vdashm \M$ and $\Gamma \lts{\lblsync{p}{q}{\ell}} \Gamma'$,
then there exists a
message label
$\ell'$, a {\ctext} $\Gamma''$ and a session $\M'$ such that $\M \lts{\lblsync{p}{q}{\ell'}} \M'$, 
$\Gamma \lts{\lblsync{p}{q}{\ell'}} \Gamma''$ and $\Gamma'' \vdashm \M'$. 
}
\end{theorem}
\begin{tcolorbox}[colback=blue!10]
\cref{lem-typ-unfold} says that typing is preserved after unfolding. 
\cref{theo-sub-red} shows that the type {\ctext} reduces along with the session it types.
\cref{theo-sess-fid} is an analogue of \cref{theo-sub-red} in the opposite direction. 
As in \cref{theo-soundness}, we allow the labels $\ell$ and $\ell'$ to be different in \cref{theo-sess-fid}.
\end{tcolorbox}
\begin{remark}\label{remark-reactive-justif}
Note that in \cref{theo-sub-red} one transition between sessions corresponds to exactly one transition
between {\ltc}s with the same label. That is, every session transition is observed by the corresponding type.
This is the main reason for our choice of 
reactive semantics (\cref{def-sess-semantics}) as $\tau$ transitions are not observed by the type in
ordinary semantics. In other words, with $\tau$-semantics the typing relation is a \textit{weak simulation} \cite{weakbisim},
while it turns into a strong simulation with reactive semantics. For our Rocq implementation 
working with the strong simulation turns out be more convenient 
\revised{as the mechanisation of weak simulation
tends to be not so straightforward \cite{stuttering_for_free,chappe_2026_family_of_sims}}. 
\end{remark}

Now we can prove two of our main results, communication safety and deadlock freedom:
\begin{theorem}[Communication Safety
    \rocqurl{STLive/lemma/live_proc}{L1285-L1288}]\label{theo-type-safe}
\revised{
Assume $\Gamma \vdashm \M$ and $\M \lts{}^* \M'$.\\
If $\M'\Rrightarrow  
(\pp \triangleleft \tout{\pq}{\ell_i}{e}.\PT  \SEP   \pq \triangleleft \pp?\{\ell_j(x_j).\QT_j\}_{j \in J}  \SEP  M'')$,
then $i \in J$.}
  
\end{theorem} 
\begin{tcolorbox}[colback=blue!10]
\cref{theo-type-safe} means that typed sessions evolve to sessions where if participant \lstin{p}
wants to send to \lstin{q} with label $\ell_i$, and \lstin{q} is listening to receive from \lstin{p},
then \lstin{q} is able to receive with label $\ell_i$.
\end{tcolorbox}
\begin{theorem}[Deadlock Freedom \rocqurl{STLive/lemma/subj_red_prog_fid}{L470}] \label{theo-progress}
\revised{If $\Gamma \vdashm \M$ , one of the following hold :
\begin{enumerate}
  \item Whenever $\M \Rrightarrow \M'$, $\M' \Rrightarrow \M_{\texttt{inact}}$ where every process making up $\M_{\texttt{inact}}$
  is inactive, i.e. 
  $\M_{\texttt{inact}} \equiv \prod_{i=1}^{n}\pp_i \triangleleft \bf{0}$ for some $n$. 
  \item Or there exists $\M'$ such that  $\M \lts{} \M'$.
\end{enumerate}
}
\end{theorem}
\begin{tcolorbox}[colback=blue!10]
\cref{theo-progress} says that the only way a typed session has no reductions available is if 
it has terminated.
\end{tcolorbox}
The final, and the most intricate, session property we prove is liveness.
\begin{definition}[Session Liveness]\label[definition]{def-live-sess}
  Let $(\rdc^* \Rrightarrow)$ denote the composition of the multistep reduction
  and unfolding relations i.e. $\N \rdc^* \Rrightarrow \N'$ iff
  $\N \rdc^* \N'' \Rrightarrow \N'$ for some $\N''$.  
  Then session $\M$ is live iff
  \begin{enumerate}
    \item ${\M} \rdc^* \Rrightarrow {\M'} =  \pq \triangleleft \tout{\pp}{\ell}{\kf{e}}.\QT \;|\; \N$ implies 
    ${\M'}\rdc^* \pq \triangleleft \QT \;|\; \N'$ for some $\N'$
    
    \item ${\M}\rdc^* \Rrightarrow {\M'} = \pq \triangleleft \sum_{i \in I} \tin{\pp}{\ell_i}{x_i}.\QT_i \;|\; \N$ implies 
    ${\M'}\rdc^* \pq \triangleleft \QT_i[v/x_i] \;|\; \N'$ for some\vspace{-2pt} 
    $\N', i, v.$ 
  \end{enumerate}
  In Rocq this is expressed with the predicate \lstin{live_sess} \rocqurl{STLive/lemma/live_proc}{L262-L270}.
\end{definition}
\begin{tcolorbox}[colback=blue!10]
Session liveness says that 
when $\M$ is live, if $\M$ reduces to a session $\M'$ containing a participant that's attempting to send 
or receive, then $\M'$ reduces to a session where that communication has happened. 
It's also called \textit{lock-freedom} in \cite{fairnesslock,padovani}.
\end{tcolorbox}
\begin{remark}
In the premises in \cref{def-live-sess}, we have used the composition 
of multistep reduction and unfolding relations.
In contrast, previous work, e.g.~\cite[Definition 2.1.3]{projsurvey}, 
defines the premise of Item 1 in \cref{def-live-sess} as 
$\M \rdc^* \M' \Rrightarrow \pq \triangleleft \tout{\pp}{\ell}{\kf{e}}.\QT \;|\; \N$. 
However, the latter definition accepts stuck sessions that 
coincidentally look like their state after taking a step.
For example, let $\M = \pp \triangleleft \mu \kf{X}.\pq ! \ell(0).\kf{X} \SEP  \pq \triangleleft \bf{0}$.
This session cannot progress and thus should not be considered live. 
By the previous definition, we have $\M \rdc^* \M \Rrightarrow \pp \triangleleft \pq ! \ell(0).\mu \kf{X}.\pq ! \ell(0).\kf{X} \SEP  \pq \triangleleft \bf{0}$.
Now the session state after the communication of $\pp$ happens is 
$\pp \triangleleft \mu \kf{X}.\pq ! \ell(0).\kf{X} \SEP  \pq \triangleleft \bf{0}$ which just happens to equal $\M$.
We further have that $\M \rdc^* \M$, hence $\M$ satisfies the liveness criteria. 
Our definition avoids this problem as we require a transition from the unfolded session 
$\pp \triangleleft \pq ! \ell(0).\mu \kf{X}.\pq ! \ell(0).\kf{X} \SEP  \pq \triangleleft \bf{0}$.
\end{remark}
\begin{theorem}[Liveness by Typing \rocqurl{STLive/lemma/fairness_feasible}{L870}]\label{theo-sess-live}
\revised{If $\Gamma \vdashm \M$ then $\M$ is live.}
\end{theorem}
\begin{proof}
\revised{We detail the proof for the send case of \cref{def-live-sess}; 
the case for the receive is similar.
Suppose that $\M \rdc^* \Rrightarrow \N$ and 
$\N \equiv \pp \triangleleft \tout{\pq}{\ell}{e}.\PT' \SEP \N'$.
Our goal is to show that there exists a $N''$ such that
$\N \equiv \pp \triangleleft \tout{\pq}{\ell}{e}.\PT' \SEP \N' \rdc^* \pp \triangleleft \PT' \SEP \N''$. 
By \cref{theo-sub-red}, we also have that  $\Gamma \vdashm \pp \triangleleft \tout{\pq}{\ell}{e}.\PT' \SEP \N'$  for some $\Gamma$.}

\revised{
Now let $\rho$ be a session reduction path starting from $\pp \triangleleft \tout{\pq}{\ell}{e}.\PT' \SEP \N'$, 
which has the following fairness property \rocqurl{STLive/lemma/live_proc}{L242-L244}: 
whenever a transition with label $\lblsync{p}{q}{\ell}$ is enabled,
a transition with label $\lblsync{p}{q}{\ell'}$ eventually occurs for some $\ell'$.
It is shown that such a path always exists, and that path can be constructed using the axioms
\lstin{constructive_indefinite_description} and \lstin{excluded_middle_informative} 
\rocqurl{STLive/lemma/fairness_feasible}{L856}. We now show that the existence of this path implies
the liveness of the session \rocqurl{STLive/lemma/live_proc}{L1699}.
}

\revised{
By extending \cref{theo-sub-red} onto paths \rocqurl{STLive/lemma/live_proc}{L41-L44}, 
let $\rho'$ be a
{\ltc} reduction path starting with $\Gamma$ 
such that every session in $\rho$ is typed by the {\ctext} at the corresponding index 
of $\rho'$, and the transitions of $\rho$ and $\rho'$ at every step match  \rocqurl{STLive/lemma/live_proc}{L246-L248}. 
Now we can show that $\rho'$ is fair  \rocqurl{STLive/lemma/live_proc}{L220-L221}.
Therefore by \cref{theo-ctx-live}, path $\rho'$ is live, 
so transition $(\pp,\pq)\ell'$ 
eventually occurs in $\rho'$ for some $\ell'$.
Therefore  
$\rho'=\Gamma \rdc^* \Gamma_0 \lts{(\pp,\pq)\ell'} \Gamma_1\rdc ..$
for some  $\Gamma_0, \Gamma_1$.}

\revised{
Now consider the session $\N_0$ typed by 
$\Gamma_0$ in $\rho$. 
We have $\pp \triangleleft \tout{\pq}{\ell}{e}.\PT' \SEP \N' \rdc^* \N_0$
by $\N_0$ being on $\rho$. 
We have that
$\N_0 \lts{(\pp,\pq)\ell''} \N_1$  for some $\ell'', \N_1$ by \cref{theo-sess-fid}.
Observe that  $\N_0 \equiv \pp\triangleleft \tout{\pq}{\ell}{e}.\PT' \SEP \N''$ for some 
$\N''$ as no transitions involving $\pp$ have happened on the reduction path to $\N_0$.
Therefore $\ell=\ell''$, so $\N_1 \equiv \pp \triangleleft \PT' \SEP \N''$ for some $\N''$, as needed.
}
\end{proof}
\begin{tcolorbox}[colback=blue!10]
\cref{theo-sess-live} shows that typable sessions are live. 
\end{tcolorbox}

\section{Related Work}
\vspace{-3pt}Examinations of liveness, also called \textit{lock-freedom}, guarantees of multiparty session types abound in the literature, e.g. 
\cite{padovani_typing_2014,LessIsMore,LessIsMoreRevisited,barbanera_partially_2023}.
Most of these papers use the definition of liveness proposed by Padovani \cite{padovani}, 
which does not make the fairness assumptions that characterise the property \cite{francez_fairness_1986} explicit. 
Contrastingly, van Glabbeek et al.\ \cite{fairnesslock} examine several notions of fairness and the liveness properties
induced by them, and devise a type system with flexible choices \cite{castellani_reversible_2019} that captures
the strongest of these properties, the one induced by the \textit{justness} \cite{fairness} assumption.
In their terminology, \cref{def-live-sess}
roughly corresponds to liveness under \revised{full fairness}, 
which is the weakest of the properties considered in that paper. They also show that their 
type system is complete, i.e. every live process can be typed. We haven't presented any completeness results 
in this paper. 
Fairness assumptions are also made explicit in recent work by Ciccone et al.\ \cite{ciccone_fair_2024,ciccone_2022-binary}, 
which use generalised inference systems with coaxioms \cite{ancona_generalizing_2017} to characterise 
\textit{fair termination}, which is a stronger property than \cref{def-live-sess}, but enjoys good compositionality properties.

Mechanisation of session types in proof assistants is a relatively new effort. 
Our formalisation is built on recent work by Ekici et al.\ \cite{srpaper}, which uses a coinductive representation
of global and local types to prove subject reduction and deadlock-freedom. 
Their work uses a typing relation between global types and sessions while ours uses one 
between associated {\ltc}s and sessions. This necessitates the rewriting of
subject reduction and deadlock-freedom proofs in addition to the novel operational correspondence, safety and liveness
properties we have proved.

\revised{In other work, Ekici and Yoshida \cite{ekici_completeness_2024}
present a mechanisation of the completeness of asynchronous subtyping.
Tirore et al. \cite{tirore2023sound,tirore_thesis, tirore2025multiparty}
give the Rocq mechanisation of a projection function for coinductive plain merge projection,
and a proof of subject reduction for an asynchronous $\pi$-calculus.}

\revised{Implementations of session types that are more geared towards 
practical verification include Castro-Perez et al.'s Zooid \cite{zooid}, which is a DSL that supports
the extraction of verified protocols via asynchronous MPST.
Similar to our work, Zooid mechanises LTS semantics for {\ltc}s 
and global types, and proves an operational correspondence between a global type and its projections.
They then exhibit that the trace of a well-typed session can be embedded in the trace of 
the global type that types it.
Unlike our work, Zooid does not formalise the properties
of the type {\ctext}s, and does not prove 
type safety and other properties of processes with respect
to a typing system. 
In contrast, our implementation directly certifies that typable
sessions are communication safe, deadlock-free and live.} 
   
\revised{The Actris framework \cite{hinrichsen2019actris} enriches
the separation logic of Iris \cite{jung2018iris} with binary session types to allow reasoning about message passing programs.
LinearActris \cite{jacobs2024deadlock} further employs linearity to certify deadlock-free programs.
Jacobs et al.'s Multiparty GV \cite{jacobs2022multiparty}, based on the functional language of Wadler's GV \cite{wadler2012gv},
proves deadlock-freedom of interleaved sessions using multiparty session 
types and by reasoning on the communication topology 
of different sessions. None of these works address liveness.
In general, verification of liveness properties, with or without session types, in concurrent separation logic is 
an active research area
that has produced tools such as TaDa \cite{gardner2021tada}, FOS \cite{lee2023fos} and LiLo \cite{lee2025lilo}
in the past few years.}

Castro-Perez et al.\ \cite{castro2026synthetic} devise a multiparty session type system 
that dispenses with projections and local types by defining the typing relation directly on the LTS
specifying the global protocol, and formalise the results in Agda. 
Li and Weis \cite{li_2025_implement} present a Rocq formalisation of 
a characterisation of \textit{implementability} for asynchronous global protocols 
given by a top-down specification. Implementable global protocols are those
corresponding to a set of deadlock-free locally specified processes, which is guaranteed in our work 
by the existence of an associated global type.
Ciccone's PhD thesis \cite{ciccone2023concertogrossosessionsfair}
presents an Agda formalisation of fair termination for binary session types. 
Binary session types were also implemented in Agda by Thiemann \cite{thiemann2019} and in Idris by Brady \cite{brady_type-driven_2017}. Several implementations of
binary session types are also present for Haskell
\cite{dardha2021,lindley2016embedding,pucella2008haskell}.
\revised{None of the
above works formalises liveness.}

\vspace{-2pt}\section{Conclusion and Future Work}
In this work we have mechanised the semantics of local and global types, proved a correspondence between them,
and used this correspondence to prove safety, deadlock-freedom and liveness for the typed sessions in 
a simple message-passing calculus. To our knowledge, our liveness result is the first mechanised
one of its kind; it is the most challenging of the theorems we formalised. 
Our implementation 
illustrates some of the difficulties encountered when mechanising liveness properties in general.
These include the use of mixed inductive-coinductive reasoning and the absence of a clear general proof technique.

\revised{A proof method for liveness in multiparty session
  processes using association was proposed in \cite{YHK2026},
  though it was proven by pencil-and-paper based on the inductive  
  full-merging projection. This association technique offers a general
proof method 
for establishing type soundness of the top-down approach, extending
to crash-stop failures \cite{BHYZ2025}
and asynchronous subtyping \cite{PY2025}. 
The present paper provides a rigorous Rocq formalisation of the type
system based on the association relation for the first time.}
The mechanisation of the association is not trivial: 
for instance, the induction on the g-context height used in the
proof of \cref{theo-ctx-live} 
requires a careful setup. The proof proceeds smoothly after this setup,
as the $\pp$-grafting of a global type neatly encodes information about the enabled 
transition of $\pp$. 
Our work further demonstrates the power of parameterised
coinduction in the verification of liveness properties, and provides a framework
for the verification of further linear time properties on session types. 

\revised{\textbf{Future Work}. There are a couple avenues for extension for our work. 
One apparent addition would be to extend the library to support coinductive full-merge projection \cite{projsurvey}.
None of the proofs included in this work explicitly use the properties of the merge operation;
therefore we expect that if the results of \cite{srpaper} can be generalised to full-merge, so can our work. 
Another possible variation on our work would be to adapt it to formalise liveness of different process calculi such as the
$\pi$-calculus, as done on paper in \cite{LessIsMoreRevisited}. The design of our proofs provides clear layers of separation between
the association of {\ltc}s to global types, which makes the {\ltc}s well-behaved, 
and the typing relation of sessions by {\ltc}s, which makes the sessions well-behaved. Therefore, we conjecture 
that any process calculus that satisfies \cref{lem-typ-unfold}, \cref{theo-sub-red} and \cref{theo-sess-fid}
could have its liveness analysed within our framework without significantly changing the remaining proofs. 
However, our requirement that the transitions match one-to-one in subject reduction and session fidelity (\cref{remark-reactive-justif})
is probably too stringent for most calculi, hence the proofs would
have to consider weak simulation relations.}


\revised{
Another possible extension would be to vary our fairness assumptions and 
examine the liveness properties induced by them, as done in \cite{fairnesslock}. 
In this paper the liveness property we targeted for our sessions is called 
\textit{Padovani lock-freedom} or \textit{lock-freedom under full fairness} in \cite{fairnesslock}.
However, we could have targeted stronger liveness properties parameterised by 
other fairness assumptions such as \textit{fairness of instructions} or
\textit{justness}. As in \cite{fairnesslock}, targeting a different liveness property 
would likely require us 
to use a different formulation of global types and a different
projection relation, but the code we used to reason about the linear-time properties of the 
LTS of sessions could be re-used.  
}

\textbf{Declaration.} We confirm that no AI generated text or code is present in this work.
\newpage

\bibliography{references}

\end{document}